\documentclass[a4paper, authoryear, 11pt]{article}

\usepackage[english]{babel}
\usepackage{pifont}
\usepackage{color}
\usepackage{babel}
\usepackage{caption}
\usepackage{tikz}
\usetikzlibrary{positioning}
\usepackage{fontenc}
\usepackage{multirow}
\usepackage{longtable}
\usepackage{graphicx}
\usepackage{tabularx}
\usepackage{natbib}
\usepackage[colorlinks=true,citecolor=red]{hyperref}
\usepackage{lineno}
\usepackage[usenames,dvipsnames]{pstricks}
\usepackage{amsmath,amsthm, amsfonts, amssymb, mathrsfs,hyperref}
\usepackage{epsfig}
\usepackage{amsthm}
\usepackage{multicol}
\usepackage{authblk}
\usepackage{float}
\usepackage{subfigure}
\numberwithin{equation}{section}
\newtheorem{thm}{Theorem}[section]

\newtheorem{remark}{Remark}

\title{Optimal harvesting under annuity and compound interest laws: economic-ecological trade-offs in a logistic growth model}

\author[1]{Anurag Sau\thanks{Corresponding author: anuragsau@gmail.com}}
\author[2]{Sujayan Gupta}
\author[3]{Uttam Ghosh}
\author[1]{Sabyasachi Bhattacharya \thanks{Corresponding author: sabyasachi@isical.ac.in}}
\affil[1]{Agricultural and Ecological research Unit, Indian Statistical Institute, Kolkata, India}
\affil[2]{Department of Mathematics, National Institute of Technology, Durgapur, India}
\affil[3]{Department of Applied Mathematics, University of Calcutta, Kolkata, India}

\begin{document}
\maketitle
\section*{Keywords:} Bioeconomic model, Pontryagin's maximum principle, Singular optimal control, Constant and varying discounting, Discount rate.

\section*{Abstract}
The relationship between investment policy associated with species growth profile is essential in seeking the most appropriate strategy for a policymaker. Balancing maximum profit with the sustainability of species remains a central issue in both ecological and economic contexts. This study presents a comparative analysis of two interest principles, annuity and compounding, within the framework of optimal control. Various investment policies are examined from the perspective of capital theory, incorporating concepts such as future value, accumulation function, and force of interest, each contributing to the formulation of an optimal control strategy. Our analysis is based on a one-dimensional logistic model incorporating linear harvesting. Key parameters include the species growth rate and interest rate, and optimality is evaluated with respect to these variables. Using Pontryagin’s Maximum Principle, we derive the optimal harvesting policies under both discounting laws and characterize the resulting steady state equilibria. The principal finding indicates that for species with low intrinsic growth rates and low annual interest rates, the annuity law of interest yields optimal outcomes. Conversely, for any annual interest rate, species exhibiting moderate or high growth rates maximize profit under the compound law of interest. The study also addresses maximum net revenue and optimal strategies for varying growth rates under each interest law.

\section{Introduction}
Optimal control theory has been applied broadly in ecological modeling to address the management of renewable resources, where the goal is to find an optimal solution where profits can be maximized while ensuring the sustainability of these species. Balancing between short-term economic returns from harvesting and long-term sustainability is a major challenge in the interface of ecology and resource economics. Over time, various articles employed rigorous mathematical and economic methods to justify the trade-off, evaluate effective management policies, and explain how economic assumptions—including the valuation of future monetary flows—effect ecological outcomes \citep{clark1974mathematical,clark1974optimal, clark1978renewable}. Our research delves into the crucial application of different laws of interest with varying species growth rates. This exploration is essential in bioeconomics and environmental management, as it provides insights into optimal harvesting strategies for species with varying growth rates. Some investors prefer to make a one-time capital investment, while others opt for a gradual investment approach. Our goal is to maximize profits from both of these strategies. Here, reinvestment is viewed as an economic effort applied based on strategic conditions and perspectives.\citep{haq2016human}. Economic growth is regarded as an increase in the output produced by a financial scheme over time. A capital asset is the present value of the expected net future revenue from the capital aspect theory \citep{wuttaphan2017human}. This capital aspect theory broadly applies to natural and biological resources, property, and buildings. Most capital assets are used to generate revenue for more than one year. Productive potential makes a common-property resource stock valuable. Miss allocation of resources over time leads to overexploitation of resources, and excessive use of a resource causes a loss of future benefits \citep{maunder2006interpreting}, \citep{clark1974mathematical}, \citep{haq2016human}.

The application of optimal control theory to renewable resources was initiated by C.W Clark \citep{clark1974mathematical}. In this pioneer foundation, a biological growth model, such as the logistic equation, was coupled with the economic aspect. Using Pontryagin's maximum principle, the author derived the necessary condition for optimal harvesting trajectories and explained how species growth rate, carrying capacity and discount rate influence steady-state biomass and harvest rate. In subsequent research, Clark addressed the problem as a joint bio-ecological system, where policymakers select harvesting efforts or rates to maximize profit while ensuring species sustainability \citep{clark2010mathematical}. These studies also indicate that the qualitative form of optimal policy is influenced by several biological phenomena, such as species growth rate, resilience, and economic parameters like discounting or accumulation rules applied for the future value \citep{clark2010mathematical, polasky2021discounting}. Other study, which includes deterministic analysis, considered constant pricing, continuous harvest rates and the solutions often led to either interior steady states harvesting or they can switch between maximum or zero harvesting depending on cost structure and discounting \citep{sethi2005fishery}. Reed \citep{reed1979optimal} explored the optimal control with harvesting in a stochastic population model, where he introduced the environmental stochasticity into the system. This study concludes that, for a broad class of Markovian growth models, a constant escapement policy (harvesting down to a fixed biomass threshold) is optimal \citep{dockner2004strategic}. This finding has been a pioneering invention of fisheries management as it is very simple and implementable, and is robust to a certain type of uncertainty. A couple of consequences studies have been developed based on this concept, where correlated environmental noise, price fluctuations, and risk-averse utility functions are taken into consideration \citep{sethi2005fishery, alvarez1998optimal}. These studies reveal that the optimal escapement level promotes more conservative harvesting. For the mathematical aspects, modern treatment of stochastic harvesting uses the Hamilton-Jacobi-Bellman equation and viscosity solutions \citep{alvarez1998optimal, lenhart2007optimal}. While theoretical studies are elegant, practical implementation often faces parameter uncertainty, data limitations and institutional constraints. Walters and Hilborn \citep{walters1978ecological} considered adaptive management approaches, where they considered the harvesting policy as an evolving feedback process updated through monitoring and learning. Also, Ludwig, Hilborn and Walters \citep{ludwig1993uncertainty} provided a review of historical management failures, where they established that overconfidence in precise optimal control solutions without accounting for uncertainty led to resource collapse.

Prediction of future value is essential for the development of a Bioeconomic model. The change in future value over time significantly influences computational methods \citep{clark1974mathematical}. Economic forces, such as supply and demand, reach a balance in the absence of external influences, resulting in a state known as the economic equilibrium point \citep{arrow1966economic, arrow1974general}. Economic equilibrium combines economic variables (usually price and quantity) toward which normal economic processes, such as supply and demand, drive the economy. Economic equilibrium can also be applied to other variables, such as interest rates or aggregate consumption spending. The equilibrium point represents a theoretical state of rest where all economic transactions that should occur, given the initial state of all relevant economic variables, have taken place \citep{andersen2013fisheries, arrow1966economic, arrow1974general}. The best element from a set of values is usually determined in mathematics and economics to obtain the desired solution. Selecting the optimal element from available alternatives under varying conditions is referred to as the mathematical optimization problem. One of the essential parts of the optimization problem is optimal control, where a control mechanism is obtained for a dynamical system over some time to optimize the objective function \cite{clark1978renewable}. In general, the extension of the calculus of variations, which is a mathematical optimization method for deriving control policies, is used in such phenomena. L.S. Pontryagin's maximum principle $(1962)$ is one of the important techniques used in optimal control, which is generally the extension of the classical variational techniques of Euler, Lagrange, Legendre, and Weierstrass \cite{boscain2005introduction, cardin2019pontryagin, mardanov2015pontryagin}. Although the maximum principle can play an important role, optimal theory serves as a unifying approach to simplify real-world problems. Optimal control has wide-ranging applications for finding mechanisms to control a system that meet specific optimality criteria. This type of problem necessitates a cost function that depends on both state and control variables \cite{la2015optimal, sargent2000optimal}.

In the usual compound law of interest, the interest is added to the present value to get the new principal value, where the interest rate remains unchanged throughout the total time and is not affected by other parameters \citep{hudson2000mathematical, lewin2019emergence}. In an annuity law of interest, the principal amount is invested at the beginning of each period, and the interest accumulates. The accumulation function is defined in terms of time $t$, expressing the ratio of the value at time $t$ (future value) and the initial investment (present value). It is used broadly in interest theory \cite{lewin2019emergence}. This function helps to compute a unit's growth after any period. In continuous compounding, the number of compounding periods reaches infinity. In those cases, the continuous compound interest is called the force of interest \citep{hardy1921notes, lukavc2025compound}.

This article provides guidance on selecting the appropriate interest law for optimal control problems based on species growth rate. When the compound interest law is applied, the discount rate remains constant, whereas the annuity law of interest results in a discount rate that varies over time. The article \citep{andersen2005production} demonstrated that the total revenue generated by an interest law is influenced by both the intrinsic growth rate of the species and the interest rate of the scheme. It describes how net revenue changes with different interest rates while keeping the intrinsic growth rate constant, as well as how net revenue varies with different intrinsic growth rates while holding interest rates fixed. We also want to discuss the preferred policy for the same invested capital over a specific time for balancing between the optimality and the species sustainability. This study develops a comparative conclusion between these two laws of interest, which suggests that it is more beneficial for small business owners compared to those with larger capital. Furthermore, we have analyzed the logistic model with linear harvesting, applying an economic perspective to identify scenarios that yield optimal profits.

The paper is organized as follows: Section \ref{Formation of interested scheme} discusses the development of future value, the force of interest, and the discount rate of interest for both compound interest and annuity law. It also provides a comparison of these laws. Section \ref{Application of different law of interest} examines the application of these interest laws within a species growth model and provides comparative studies aimed at maximizing profit. We provide a detailed discussion on bio-economic equilibrium and singular optimal control within this model. Finally, we summarize our findings in section \ref{conclusion}.

\section{Interest accumulation framework: Accumulation function, future value, and the force of interest} \label{Formation of interested scheme}

One of the essential aspects of financial theory is principal accumulation, primarily used to illustrate the relationship between time and money in any economic scheme. In the real world, selecting an appropriate interest law with an increasing accumulation function is crucial, as profit generation remains a primary objective. Various interest laws, including simple interest, compound interest, and annuities, can be applied; however, compound interest is often preferred due to its efficacy in maximizing returns over time \citep{lewin2019emergence}. Calculating future value will aid in selecting the most beneficial interest law for a specific situation. From a modeling perspective, it will also help us optimize the wage, rent, and other additional costs. Consideration of cash flows and their timing is essential when establishing the time value relationship \citep{clark1974optimal}. The accumulation function for any interest law can be constructed based on five key variables: present value, future value, time interval, initial payment amount, and interest rate. The computational methodology was provided by Hudson \citep{hudson2000mathematical}. The interest rate depends on different interest laws, the model's needs, and the initial value \citep{samuelson1937some}. 
 
From an economic perspective, the accumulation function is expressed in terms of time $t$, representing the ratio of the value at time $t$ (future value) and the initial investment (present value) \citep{vaaler2021mathematical}. This concept is widely used in interest theory to model the growth of investments or other quantities over time. Commonly, the accumulation function is generated by accumulating the variations of the value that occur at some rate over moments of its independent variable \citep{thompson2008concept}. The accumulation function describes the process through which capital, wealth, or resources change due to savings, investments, or reinvestment. We define capital accumulation as the growth of physical assets such as machinery, buildings, and technology, producing more goods and services, and leading to economic expansion. Wealth accumulation occurs through saving and investing, driven by the effects of compound interest and capital gains.

The accumulation function is denoted as $\alpha (t)$, and it takes various forms for the different schemes. However, it adheres to the general property that, for any time interval $n$  ($n\geq0$), $\alpha (0)=1$ \citep{clark1974mathematical}. In the case of continuous compounding, as the number of compounding periods $n$ reaches infinity, we refer to continuous compound interest as the force of interest $\delta$. The force of interest may be constant or vary with time, depending on the law of interest. For any continuously differentiable function $\alpha(t)$, the force of interest, often called the logarithmic or compounded return, is a function of time defined as follows 
\begin{center}
 $\delta (t)=\frac{\alpha'(t)}{\alpha(t)}$=$\frac{d}{dt} [\ln {\alpha(t)}]$,
\end{center}

where $\alpha'(t)$ represents the derivative of the accumulation function with respect to time $t$. 
In financial theory, the force of interest, $\delta (t)$, plays a significant role in describing how rapidly the value of an investment grows continuously. A constant value of $\delta(t)$ indicates a uniform continuous compounding rate. Conversely, if $\delta(t)$ changes over time, it represents more complex financial arrangements, such as those described by the annuity law of interest \citep{clark1974mathematical}.

\subsection{Compound interest: fundamental formulation}\label{General results for the compound scheme}
\begin{longtable}[h] {| p{.1\textwidth} | p{.30\textwidth} | p{.60 \textwidth} |}
\hline
 Parameter & Name & Description \\
 \hline
 $t$  & Time & Duration of the investment  \\
\hline
 $\alpha(t)$  & Accumulation function & Obtained by expressing the ratio of the value at time t (future value) and the initial investment (present value)  \\
\hline
$\alpha' (t)$  & Rate of change of accumulation function over time  & Derivative of the accumulation function with respect to time $t$ \\
\hline
$n$ & Intervals  & Number of intervals of payment or, time intervals \\
\hline
$i$  & Annual rate of interest  & The amount is invested at a interest rate $i$ \\
\hline

$\delta(t)$  & Force of interest & Defined as the nominal rate of interest payable continuously   \\
\hline
$P^*_0$ & Initial amount  & Initial amount invested at compound interest \\
\hline
$P^*_i$  & Investment amount  & Investment after $i$ number of  years  \\
\hline
 $P_0$ & Principal amount  & Amount invested in an annuity scheme for n years with a rate of interest $i$ at initial stage\\
\hline
$P_n$  & Value of the principal amount after n years  & Value of the invested amount  $P_0$  after $n$ years  \\
\hline
$PV$  & Present value  & The present value of a future payment at time t is   \\
\hline
 $P$ & Future payment  & Future payment based on $PV$  \\
\hline
$p$  & Selling price & The selling price of per unit harvested population   \\
\hline
$c$ & Cost per unit efforts  & The constant harvesting cost per unit efforts  \\
\hline
$\lambda$  &  Parameter of the model & Adjoint parameter of the model \\
\hline
$r$ & Intrinsic growth rate  & Intrinsic growth rate of the species\\
\hline
\caption{The list of parameters used in this article.}
\end{longtable}

The initial investment, commonly known as the principal, represents the minimum capital invested at the outset and serves as the basis for calculating interest, fees, and returns. The growth of this investment is analyzed compounding formulas to determine how the amount changes based on the invested capital, the duration of the investment, and the interest rate. Estimating the future value allows for the assessment of the worth of a current payment at future. The difference between the future value and the initial value indicates the profit or loss generated by a particular investment strategy. Essentially, future value takes a current sum of money and projects its worth at a specified time in the future, while present value determines what a future sum of money is worth today \citep{shackle1965scheme}.

The compound law of interest has significant applications in investment policies and resource management \citep{lewin2019emergence}. Let's assume that an initial payment $P^*_0$ is invested at compound interest, where $i$ represents the rate of interest. In this case, the future value increases exponentially, and after $n$ years, it will be \citep{clark1974mathematical}
\begin{center}
	Future value = $P^*_0(1+i)^n$.
\end{center}
\noindent Now, assume $\delta = \ln {(1+i)}$. By substituting this into the earlier expression, we can generalize the future value as:
\begin{center}

	Future value= $P^*_0 e^{\delta t}$,

\end{center}
for an arbitrary time point $t\geq 0$. Here, $\delta$ is often referred to as the annual rate of interest (compounded continuously/force of interest). The process of discounting, which refers to the valuation of future payments, is the reverse of compounding interest on present payments. Therefore, the present value of a payment $P^*_t$ at $t$ years is 
 \begin{center}
 	Present value= $P^*_t e^{-\delta t}$.
 \end{center}
 Here, $\delta$ represents the instantaneous rate of discount. The terms discount and interest rates are used interchangeably when discussing future payments being discounted or present payments being compounded, respectively. If $P^*_0$, $P^*_1$, $P^*_2$,....,$P^*_N$ is invested in $0, 1, 2,..,N$ years respectively, then the present value is expressed as: 
 
 \begin{center}
  PV= $\displaystyle \sum_{k=0}^{N}\frac{P^*_k}{(1+i)^k}$.	
 \end{center}
  Similarly, for a continuous time stream of revenues $P^*(t)$, $0\leq t \leq T$, given by 
 
  \begin{center}
  	  PV=  $\displaystyle \int_{0}^{T} P^*_t(t)e^{-\delta t}$,
  	 \end{center} 
where $T$ represents the time interval for generating revenue.
 
 \subsection{Annuity interest: time-dependent accumulation and a new paradigm of thought process} \label{Formulation of general results for annuity scheme} 

Let $P_0$ be the principal amount invested in the annuity law of interest at the initial time, and $i$ be the annual rate of interest. After one year the principal value will be $P_1$ =$P_0 + iP_0 = P_0(1+i)$.\\
Similarly, the principal amount after the second year will be
\begin{align*}
	&P_0(1+i)+P_1(1+i)\\  
	&=\frac{P_0}{i} (1+i)[(1+i)^2-1] 
\end{align*}
Succeeding if we calculate the principal after $n$ years, the mathematical expression can be obtained by mathematical induction.

\begin{thm}
    If $P_0$ amount is invested in an annuity scheme with a rate of interest $i$ for $n$ years, then the amount $(P_n)$ after $n$ years can be calculated using the simple formula $P_0(1+i)\frac{(1+i)^n-1}{i}$. 
\end{thm}
\begin{proof}
 Consider, $n$ =1, then $P_1=P_0(1+i)$. So, the equality is true for $n=1$.
Let us assume the equality is true for $n=m$. Now,
\begin{align*}
P_{m+1} & =(P_0+P_m)+(P_0+P_m)i\\
& =P_0(1+i)+P_0 \frac{(1+i)((1+i)^m-1)}{i}(1+i)\\
&=P_0(1+i) \frac{(1+i)^{m+1}-1}{i}
\end{align*} 
\noindent By mathematical induction, if the statement is true for $n=m$, it will also be true for $n=m+1$. Hence, the equality is true for any natural number.
\end{proof}

\noindent The accumulation function and the force of interest for the annuity law of interest are derived as follows. 
Let $\alpha(t)$ denote the accumulation function, which is expressed as $\alpha(t)=(1+i)\frac{(1+i)^t-1}{i}$.
Therefore, for the annuity law of interest, we have,
\begin{align*}
	\delta(t)&=\frac{\alpha'(t)}{\alpha(t)}\\
	    & = \frac{(1+i)^t\ln{(1+i)}}{(1+i)^t-1}\\
        &= (1+r_1^{-t})\ln {r_1}~\text{[Let $(1+i)=r_1$, neglecting higher order terms as $(1+i)>1$, and $(1+i)^{-t}<1$)].} 
	\end{align*}
\noindent According to the annuity law of interest, the present value of a future investment is determined by inverting the force of interest. This calculation utilizes the accumulation function.

\begin{thm}
	Let the instantaneous discount rate, denoted as $\delta$, be a function of time such that $\delta = \delta(t)$. The present value of a future payment $P$  at time $T$ is then expressed as $PV = Pe^{-\int_{0}^{T} \delta(s) ds}$.
	 \end{thm}
 \begin{proof}
 	Let $v(t)$ denote the value of the deposit at time point $t$ \citep{clark1974mathematical}. Then we have,
 \begin{eqnarray*}
 \frac{dv(s)}{dt}=\delta (t)v(s)\\
 \implies \frac{dv(s)}{v(s)}= \delta(t) dt
 \end{eqnarray*} 
 By integrating from $t=0$ to $t=T$, we get
 \begin{align*}
     \int_{0}^{T} \frac{dv(s)}{v(s)} &= \int_{0}^{T} \delta(t) dt\\
 \implies v(0) &= v(t)e^{- \displaystyle \int_{0}^{T} \delta(t) dt} 
 \end{align*}
 Now, $v(0)=PV$, $v(t)=P$, therefore $PV= Pe^{-  \displaystyle \int_{0}^{T} \delta(t) dt}$  	
 \end{proof}

\begin{remark}
It is observed that the discount (interest) rate varies with time under the annuity law of interest. Therefore, the preceding result can be applied to determine the present value according to this law. The present value for the annuity law of interest is given as follows:

\begin{align*}
& P_t\displaystyle \int_{0}^{\infty} e^{- \displaystyle \int_{0}^{t}\delta(s)ds} dt\\
& =P_t\displaystyle \int_{0}^{\infty} e^{- \displaystyle \int_{0}^{t} \ln{r_1} (1+r_1^{-s})ds} dt\\
& =P_t \displaystyle \int_{0}^{\infty} e^{-t\ln {r_1}+r_1^{-t}-1}dt 
\end{align*}
\end{remark}
Under the compound law of interest, the force of interest remains constant throughout the time period. In contrast, the annuity law of interest is time-dependent, exhibiting a decreasing function that stabilizes and becomes nearly constant after a specific interval.
The changes in the discount rate for both the compound and annuity law of interest are provided in the figures \ref{exp_delta_t_compound} and \ref{exp_delta_t_annuity}, generated for various interest rates. For the compound law, the discount rate approaches zero regardless of the interest rate. The annuity law similarly demonstrates exponential decay. As depicted in both Figures \ref{exp_delta_t_compound} and \ref{exp_delta_t_annuity}, increasing the interest rate causes \(e^{-\delta t}\) to approach zero. This outcome suggests that a higher interest rate requires a smaller present value to achieve maximum profit \citep{clark1974mathematical}, \citep{pascoe2002optimal}, \citep{barua2019maximum}.



\begin{figure}[H] 
    \centering
\includegraphics [height = 100mm, width =150mm]{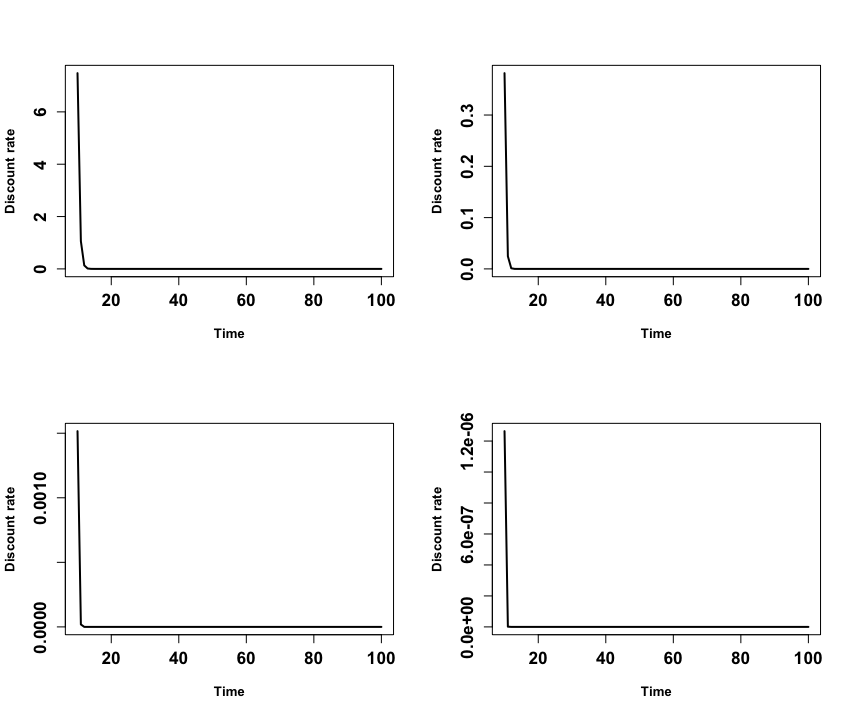}
    \caption{This figure illustrates the change of discount rate ($e^{- \delta t}$) for different interest rates ($i$) under the compound law of interest. The interest rates considered are $0.035, 0.055, 0.085$, and $0.115$. We scaled the y-axis by multiplying all values by $10^7$ units.}.
    \label{exp_delta_t_compound}
\end{figure}

\begin{figure}[H] 
    \centering
    \includegraphics[height = 100mm, width =150mm]{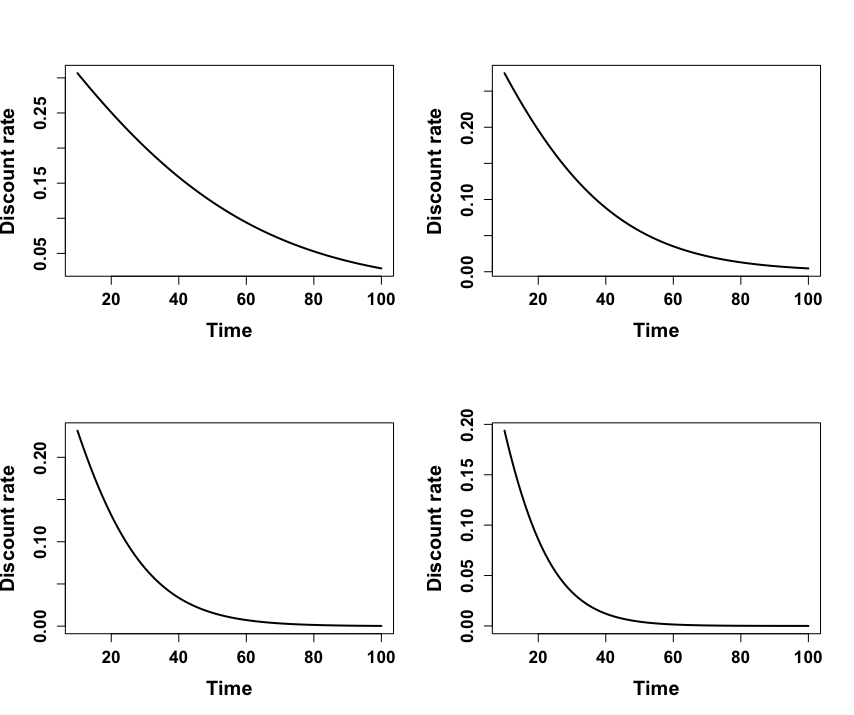}
    \caption{This figure illustrates the change of discount rate ($e^{- \delta t}$) for different interest rates ($i$) under annuity law of interest. The different interest rates are taken from the previous figure.}
     \label{exp_delta_t_annuity}
\end{figure}

%
%
%
%

\section{Optimal control of logistic growth under different law of interest} \label{Application of different law of interest}

This section focuses on developing an optimal harvesting strategy by comparing two laws of interest in the context of logistic growth with linear harvesting. The economical valuable  species are being harvested for commercial purposes. Implementing an optimal control policy enables effective balancing of the trade-off between species sustainability and profit maximization \citep{li2017bioeconomic}. The linear or proportional harvesting strategy is determined by the catch-per-unit effort hypothesis \citep{clark1974mathematical} and expressed as $H(x, E)= qEx$, where $q$ represents the catchability coefficient, and $E$ is the harvesting effort. Here, $H$ is an unbounded function of one variable of either $E$ or $x$, where the other variable is constant.

 We considered the general logistic model with linear harvesting, represented by
\begin{equation}\label{equation}
\frac{dx}{dt} = rx\left(1-\frac{x}{K}\right)-qEx ,
\end{equation}

where $r$ is the intrinsic growth rate, $q$ is the capability coefficient, $E$ is the harvesting effort, $K$ is the carrying capacity \citep{li2017bioeconomic}. 

Implementing an optimal harvesting strategy within a logistic growth model with linear harvesting is fundamental to effective bioeconomic management. The trade-off between maximizing profit and ensuring species sustainability remains a central consideration. Notably, it is crucial to select our strategies based on the intrinsic growth rate $(r)$, as the growth rate plays an important role in determining the profit. For species with a low intrinsic growth rate, the population reaches the carrying capacity gradually and tends to remain stable there \citep{foley2012review}. However, for a moderate or high intrinsic rate, the species reaches the carrying capacity faster and remains that population size. The rate of interest is also a vital parameter for maximizing the profit function. Therefore, conclusions are drawn based on these two parameters \citep{rempala1992decision}.  

\subsection{Optimal harvesting}
We aim to provide the optimal policy for a trade-off between the current and future harvest to utilise renewable resources properly. We consider the profit gained by harvesting. In our analysis, we focus on linear harvesting to derive an analytical expression and provide a proper interpretation of the results.

\subsubsection{Profit structure and bionomic equilibria}
Bioeconomic equilibrium is generally achieved when the total revenue generated by selling the harvested population equals the total cost invested in harvesting. The economic interest in the yield of the harvested effort defines the net economic revenue. The net economic revenue is given as follows:

Net economic revenue (N.E.R) is defined as total revenue (T.R) - total cost (T.C), where T.R = ($pqxE$) and T.C = $cE$ . Here, $c$ is the constant harvesting cost per unit effort and $p$ is the selling price per unit harvested population. Therefore, the net profit at any time is given by $(pqx-c)E$  \citep{srinivas2014optimal}. To satisfy the positivity condition, we have ($pqx-c > 0$), which implies that $ x > \frac{c}{pq}$ \citep{flaam1996approaches}.

\subsubsection{Singular optimal control characterization}
In this approach, our primary objective is to maximize the present value of the continuous-time revenue stream. This objective can be expressed mathematically as follows \citep{clark1974mathematical}:
\begin{equation}\label{optimal_function}
J(x,E)= \int_{0}^{\infty}e^{\delta t}P(x,E) dt  
\end{equation}
Here, $\delta$ is defined earlier, and harvesting agencies determine its value based on their preferences. We maximize equation \ref{optimal_function} together with the steady-state equation \ref{equation}. To address this maximization problem, Pontryagin's maximum principle is used \citep{gupta2012bifurcation}. In this context, $E$ represents the control variable, and it satisfies the constraint $0\leq E \leq E_{max}$, where $E_{max}$ is a feasible upper limit for the harvesting effort. Consequently, over an infinite time horizon, the optimal control problem is formulated as
\begin{center}
	\[
	\max_{0\leq E(t) \leq E_{max}} \int_{0}^{\infty}e^{\delta t}P(x,E) dt
	\]
\end{center}   
with $x(0)=x_0$.
The Hamiltonian function associated with this model can be expressed as:
\begin{eqnarray*}
	H(x,E,\lambda)=e^{-\delta t}[p(x,E)]+
	\lambda\frac{dx}{dt}\\
	=e^{-\delta t}\left[p(x,E) \right]+\lambda\left[rx\left (1-\frac{x}{K}\right)-qEx\right],
\end{eqnarray*}
where $\lambda$ is the adjoint variable \citep{raymond2019modeling}. 

Differentiating the Hamiltonian $H$ with respect to $E$ we get\\
\begin{equation}\label{differentiating_E}
\frac{\partial H}{\partial E}= (pqx-c)e^{-\delta(t) t} - \lambda q x .
\end{equation}

The above control problem confesses a singular solution on the control set $[0, E_{max}]$ when $\frac{\partial H}{\partial E}=0 $, which gives 
\begin{equation} \label{shadow_price_equation}
\lambda e^{\delta(t) t}= p-\dfrac{c}{qx},
\end{equation}   
where $\lambda e^{\delta(t) t}$ is called the shadow price \citep{gupta2012bifurcation}.
To determine the path of a singular control problem, according to Pontryagin's Maximum Principle, the adjoint variable has to satisfy the adjoint equation given by 
\begin{equation}
\frac{d\lambda}{dt}= -\frac{\partial H}{\partial x}
\end{equation} \citep{gupta2012bifurcation}.

Now, we are searching for a steady state solution $x^*$ by considering $E$ as a constant. Therefore,  

\begin{equation} \label{steady_state}
\begin{split}
    &\frac{d \lambda}{dt}=-pqEe^{-\delta(t)t}-\lambda \left[r\left(1-\frac{2x^*}{K}\right)-qE\right]\\
	 \implies & \frac{d \lambda}{dt}+\lambda \left[r\left(1-\frac{2x^*}{K}\right)-qE\right]=-pqEe^{-\delta(t)t}\\
	 \implies & \frac{d \lambda}{dt}+\alpha\lambda = \beta e^{-\delta(t)t}
	\end{split}
\end{equation}
where, 	$\alpha(x)=\left[r\left(1-\frac{2x^*}{K}\right)-qE\right]$ and	$\beta(x)=-pqE $. 

Determining the steady state of equation \eqref{steady_state} is challenging because of the presence of $e^{-\delta (t)t}$. To address this issue, the transformation $\lambda(t)=\mu(t) e^{-\delta(t)t}$ is applied. Consequently, the equation can be reformulated as shown in \citep{raymond2019modeling},
 \begin{eqnarray} \label{transversality_equation}
\begin{split}
\frac{d \mu (t)}{dt}-\delta(t)\mu (t)+\alpha \mu (t)=\beta\\
	\implies \frac{d \mu (t)}{dt}+[\alpha - \delta(t)]\mu (t)=\beta\\ 
\end{split}
\end{eqnarray} 

To satisfy the transversality condition at $\infty$, $\lim_{ t \to \infty} \lambda (t)=0$, the shadow price $\mu (t) = \lambda e^{\delta(t)t}$ should remain constant over the line in singular equilibrium. Thus, the solution of the equation \ref{transversality_equation} satisfying the transversality condition is given by    
\begin{equation} \label{value_beeta}
	\mu (t)=\frac{\beta}{\alpha - \delta(t)}\\
\end{equation}
Then from the equation \ref{shadow_price_equation} and and \ref{value_beeta}, we get
\begin{eqnarray} \label{final}
	\begin{split}
	& \frac{\beta}{\alpha - \delta(t)}=p-\frac{c}{qx^{*}}\\ 
	\implies & \frac{-pqE}{r(1-\frac{2x^{*}}{K})-qE}=p-\frac{c}{qx^{*}}\\ 
	\implies & \frac{2pq}{K} {x^*}^{2}-\left(pq+\frac{2c}{K}\right)x^{*}+c\left(1-\frac{qE}{r}\right)=0\\ 
\end{split}
\end{eqnarray}	
	
According to Descartes' rule of signs, the equation may yield two positive roots or one positive and one negative root, depending on certain constraints. If $\left (1-\frac{qE}{r}\right)>0$, the equation will have two positive roots and if $\left(1-\frac{qE}{r}\right)<0$, it will have one positive root and one negative root \citep{berck1979open, bhattacharyya2017macroalgal}. By combining equations \ref{equation} and \ref{final}, we the singular point is given by $ \left(x^*, E^*\right)= \left(x^{*},\frac{r}{cq}(\frac{2x^{*}}{K}-1)(pqx^{*}-c)\right )$.

\subsection{Result and discussion}\label{Result}

Previous research on optimal control has primarily addressed general strategies based on the compound law of interest. Our study examines both the compound and annuity laws of interest (constant discounting and varying discounting) and provides a comparative study on the application of optimality. We aim to provide insights into which economic aspect is better in terms of profit. The logistic model with linear harvesting is employed to compare both laws of interest and provides tabular representations and figures to facilitate a comprehensive comparison study. The findings indicate that the revenue depends on the intrinsic growth rate ($r$) and the rate of interest ($i$) for the linear harvesting. Distinct  optimal strategies for different $r$ and $i$, demonstrating the practical implications of our findings for decision-making in economic modeling. Both the compound interest and annuity law of interest are relevant across different intrinsic growth rates, leading to various optimal strategies for different combinations of $r$ and $i$.

Figure \ref{total_revenue}(a) indicates that at the initial stage, the net revenue generated by the annuity law of interest is higher than that of the compound law of interest for varying intrinsic growth rates. However, beyond a certain threshold, the net revenue from the compound law surpasses that of the annuity law. This profit also depends on the rate of interest. When both the growth rate and interest rate are low, the annuity law of interest (continuous discounting) yields the maximum profit. Conversely, at higher interest rates, the compound interest yields greater profit \ref{total_revenue}(b). For species exhibiting significant growth rates, the compound law of interest consistently provides the highest profit, as illustrated in Figure \ref{total_revenue}(c).

Under the annuity law of interest, the net revenue is achieved over a longer time frame when the intrinsic growth rate is low. As the value of $r$ increases, the net revenue rises; however, beyond a certain threshold, the net revenue begins to decline (as illustrated in figure \ref{net_revenue_different_r} (b)). Although an increase in $r$ leads to stabilization of net revenue at the maximum profit level, a very small intrinsic growth rate ($r=0.1$) can result in a high maximum profit that drops sharply. The net revenue is high when $r=0.1$, but as $r$ increases, it is attained at a more moderate level. After reaching its peak, profit tends to decline and eventually stabilizes over time. Initially, the compound law of interest tends to yield higher profit, but as time progresses, the profit decreases. As shown in Figure \ref{net_revenue_different_r}(a), after a certain period, net revenue exhibits stability for an extended time. For a higher growth rate, stabilization occurs earlier. In both cases, when $r$ is small ( below $0.3$), the annuity law of interest generates maximum profit, but when $r$ exceeds $0.3$, the compound interest yields greater profit. This figure illustrates that, if $r=0.7$, a net revenue of $1093$ units is reached after $29$ time periods. In contrast, this same revenue is achieved through the annuity law after only $13$ time periods. On the other hand, when $r$ is small, i.e., $r=0.35$, the net revenue of $1110$ units is generated at $29$ and $30$ time periods for the compound and annuity laws of interest, respectively. Thus, in this case, both methods yield the same profit after a comparable duration. This interpretation is more apparent from the table \ref{table_different_r} (a) and (b), which represent the compound and annuity laws of interest, respectively, with an interest rate $i$ set at $0.04$. In both scenarios, an increase in the growth rate results in a decline in the net revenue. For a comparatively low growth rate ( in our case $r=0.1, 0.3$), the annuity law of interest gives more profit, and for the moderate value of $r$ ($r=0.5,0.7$), compound interest to be more advantageous.

Figure \ref{net_revenue_different_i_small_r}(a) and (b) shows the net revenue for various interest rates ($i$) when the intrinsic growth rate ($r$) is relatively low (here $r$ = 0.3) for compound and annuity law of interest, respectively. At lower interest rates ($i=0.2,0.4$), the annuity law of interest yields higher net revenue. However, as the interest rate increases, compound interest provides more profit than the annuity law of interest. This trend is detailed in Table \ref{table_diff_i_small_r}, which provides a detailed comparison for different values of $i$ with a constant intrinsic growth rate ($r=0.3$). As anticipated, net revenue increases with higher interest rates, but the specific profit depends on the interest rate values. Shifting focus to moderate growth rates, Figure \ref{net_revenue_different_i_modarate_r} represents the net revenue for various interest rates ($i$) with a moderate intrinsic growth rate ($r=0.7$). Likewise, Table \ref{table_diff_i_modarate_r} interprets results consistent with the figure, displaying data for a fixed moderate intrinsic growth rate ($r=0.7$) across different interest rates. Notably, when the intrinsic growth rate is relatively high, the compound law of interest consistently provides the maximum profit compared to the annuity law of interest.

Harvesting effort is also an essential constraint from an economic perspective. More harvesting effort cost reduces the profit or net revenue in the optimal control problem. When the interest rate is fixed, figure \ref{total_effort_diff_r} represents the total efforts for different intrinsic growth rates ($r$). For the compound law of interest, an increase in growth rate requires more effort to achieve the optimal profit. We need more effort initially, and then we must reduce the effort. For the annuity law of interest, we have to increase the effort over time. Later, we must put in the maximum effort to get the optimal profit. For a small growth rate, optimum effort should be applied earlier than it would for a species with a significant growth rate. 

\begin{figure}[H]
	\begin{center}
		\subfigure[]{\includegraphics[height=6 cm, width=8cm]{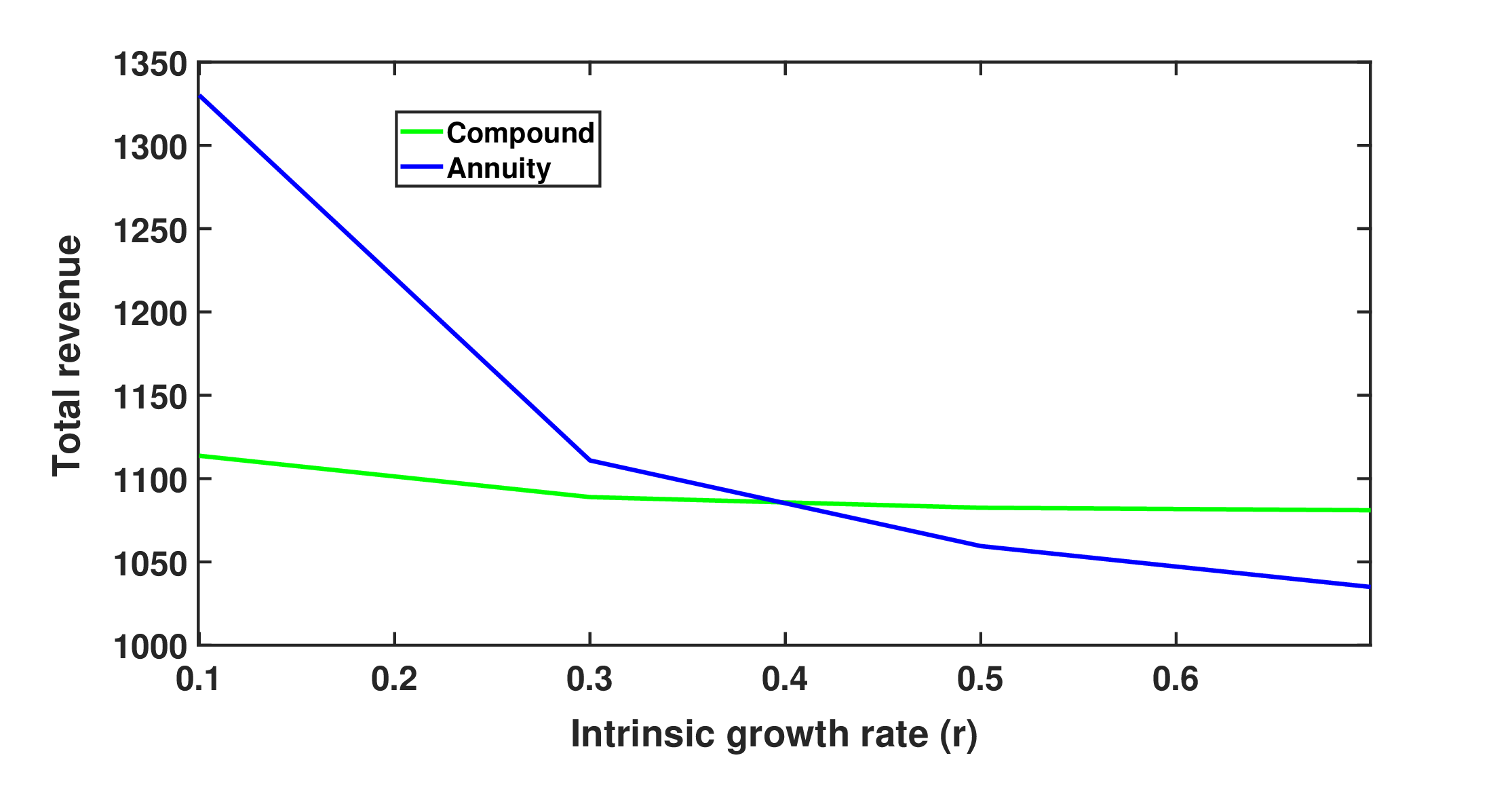}}
        \subfigure[]{\includegraphics[height=6 cm, width=8cm]{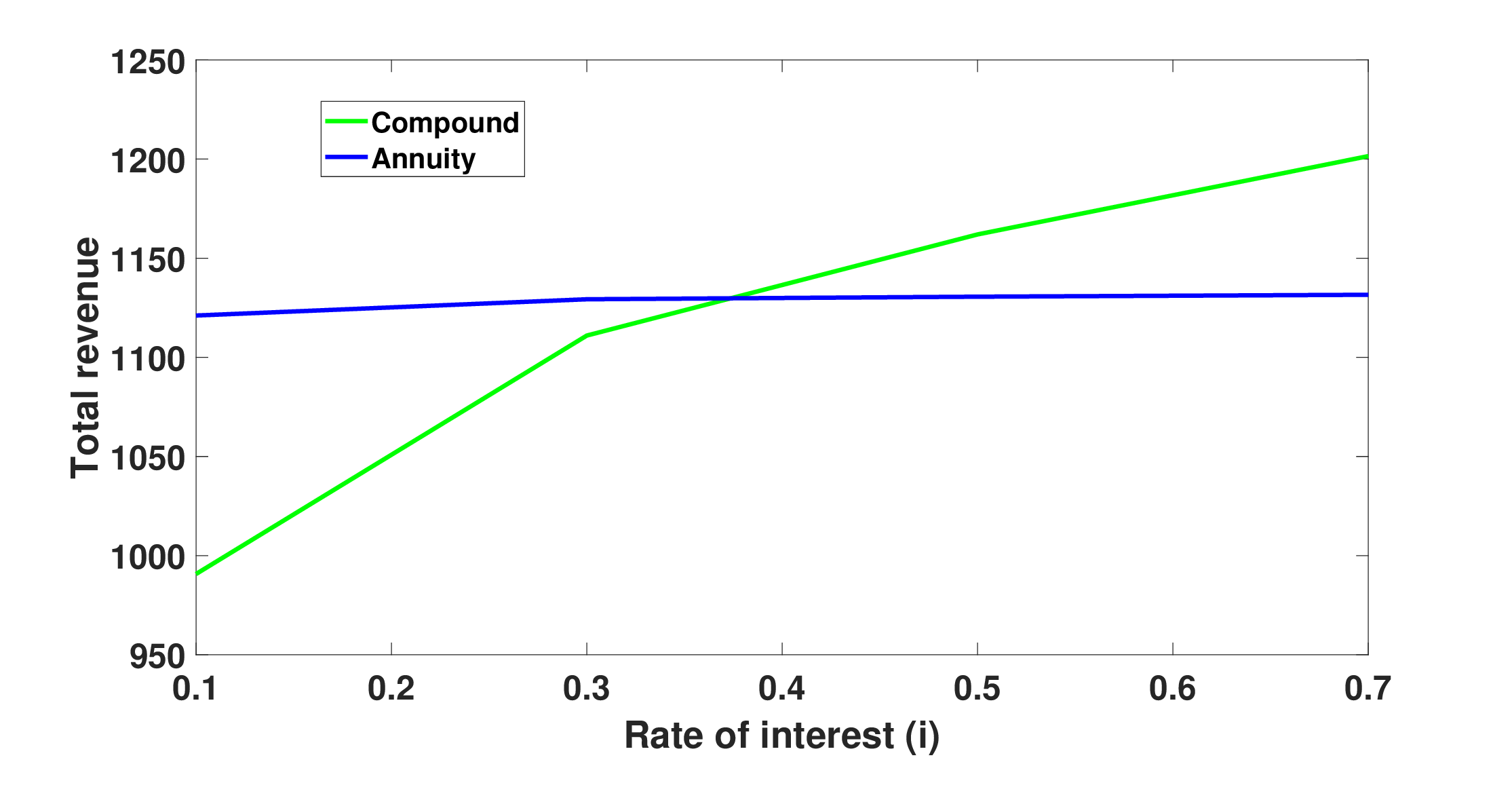}}
        \subfigure[]{\includegraphics[height=6 cm, width=8cm]{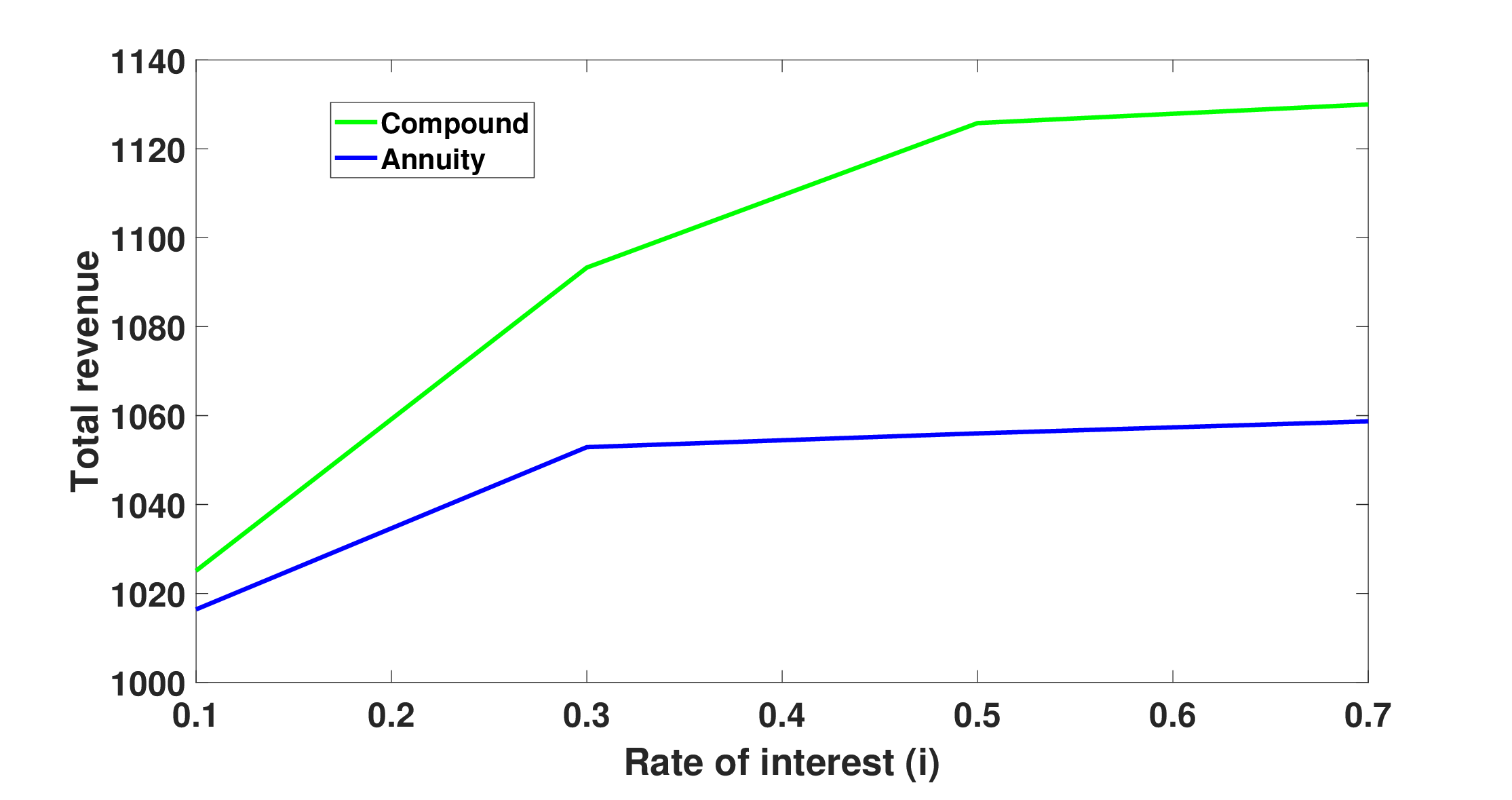}}
        
       \caption{These figures illustrate the total revenue generated under the compound and annuity laws of interest. Figure (a) demonstrates that a species with a low intrinsic growth rate achieves the optimal trade-off for the annuity law of interest. Increasing the growth rate reveals a threshold value, above which the compound law of interest yields the optimal profit. Figures (b) and (c) are plotted with respect to the rate of interest, considering growth rates below and above this threshold value, specifically $r=0.3$ and $0.7$).}
  \label{total_revenue}
	\end{center}

\end{figure}

\begin{figure} [H]
	\begin{center}
		\subfigure[]{\includegraphics[height=6 cm, width=8cm]{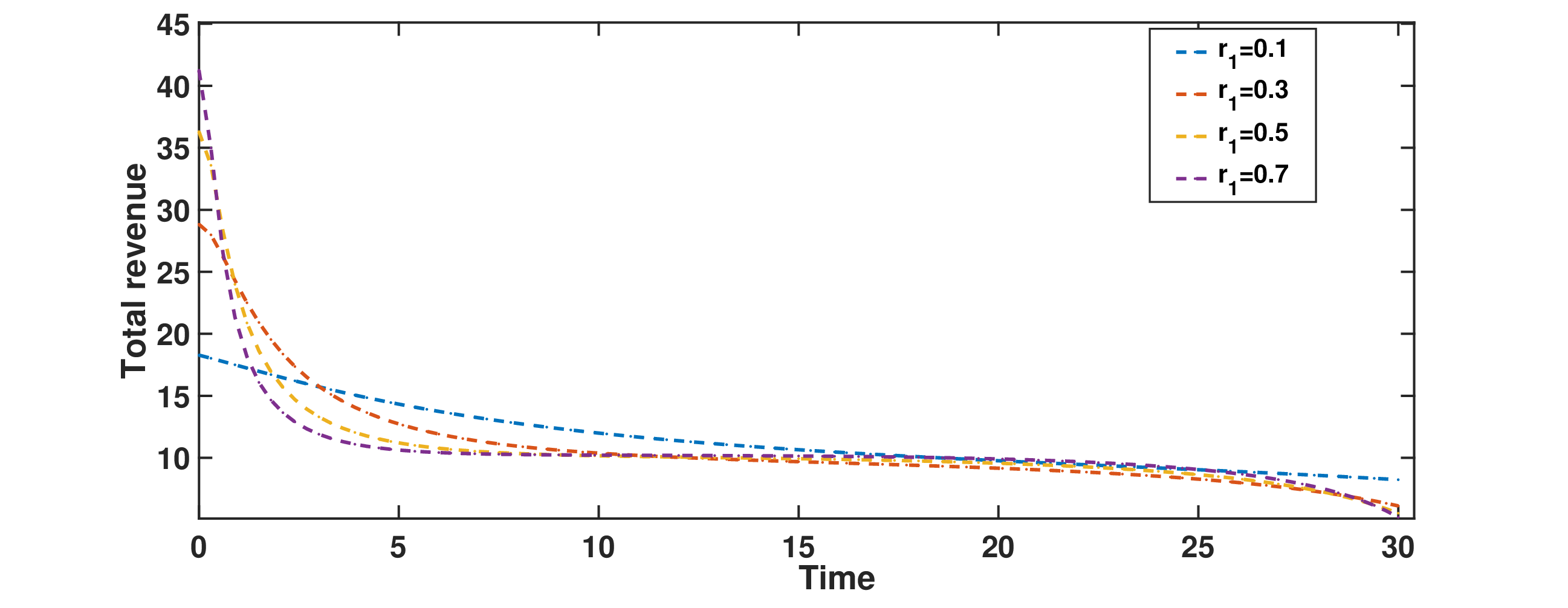}}
		\subfigure[]{\includegraphics[height=6 cm, width=8cm]{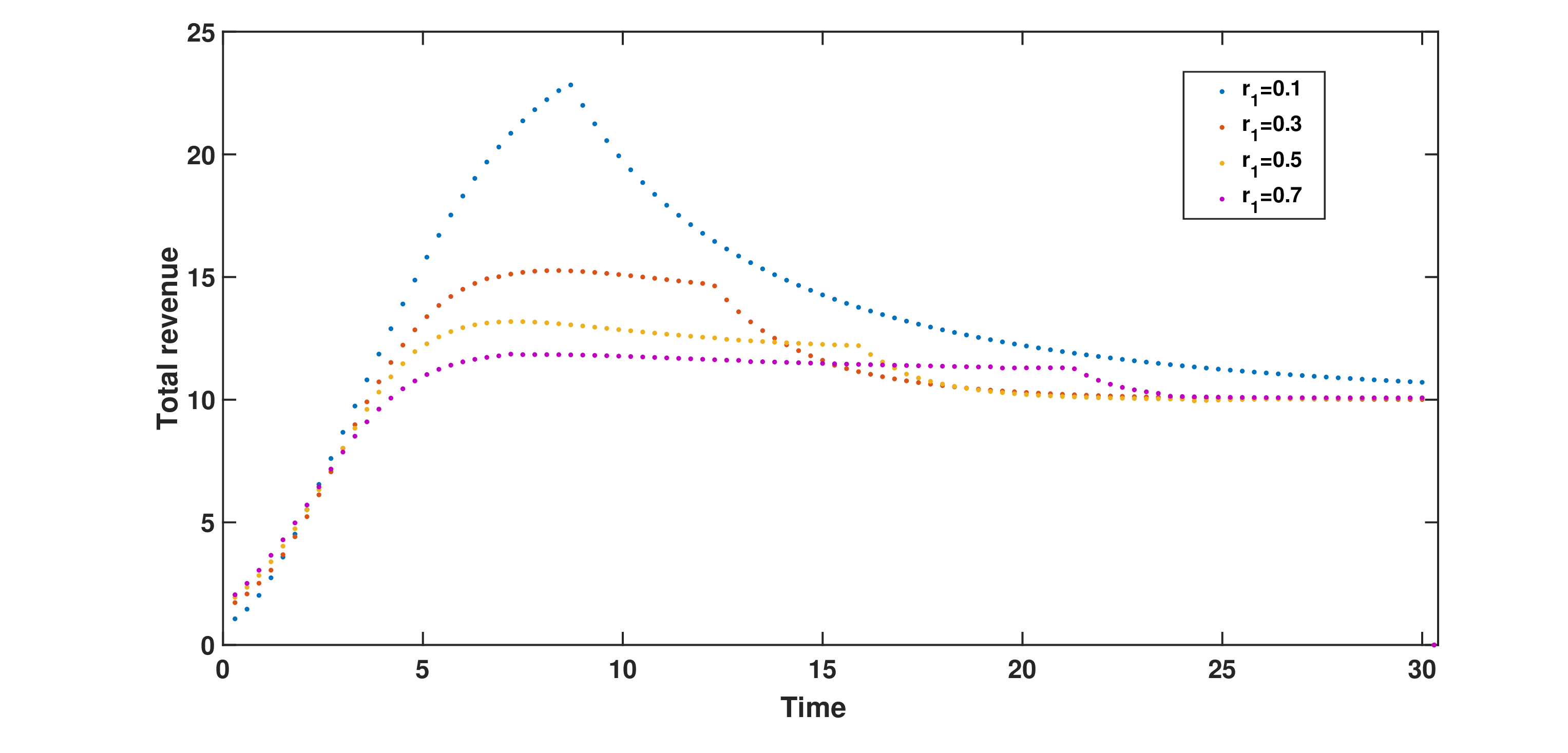}}
		\caption{These figures represent the total net revenue produced by varying intrinsic growth rates under compound and annuity interest laws, respectively.}
  \label{net_revenue_different_r}
		\end{center}
  
\end{figure}

\begin{figure}[H]
	\begin{center}
		\subfigure[]{\includegraphics[height=7 cm, width=8cm]{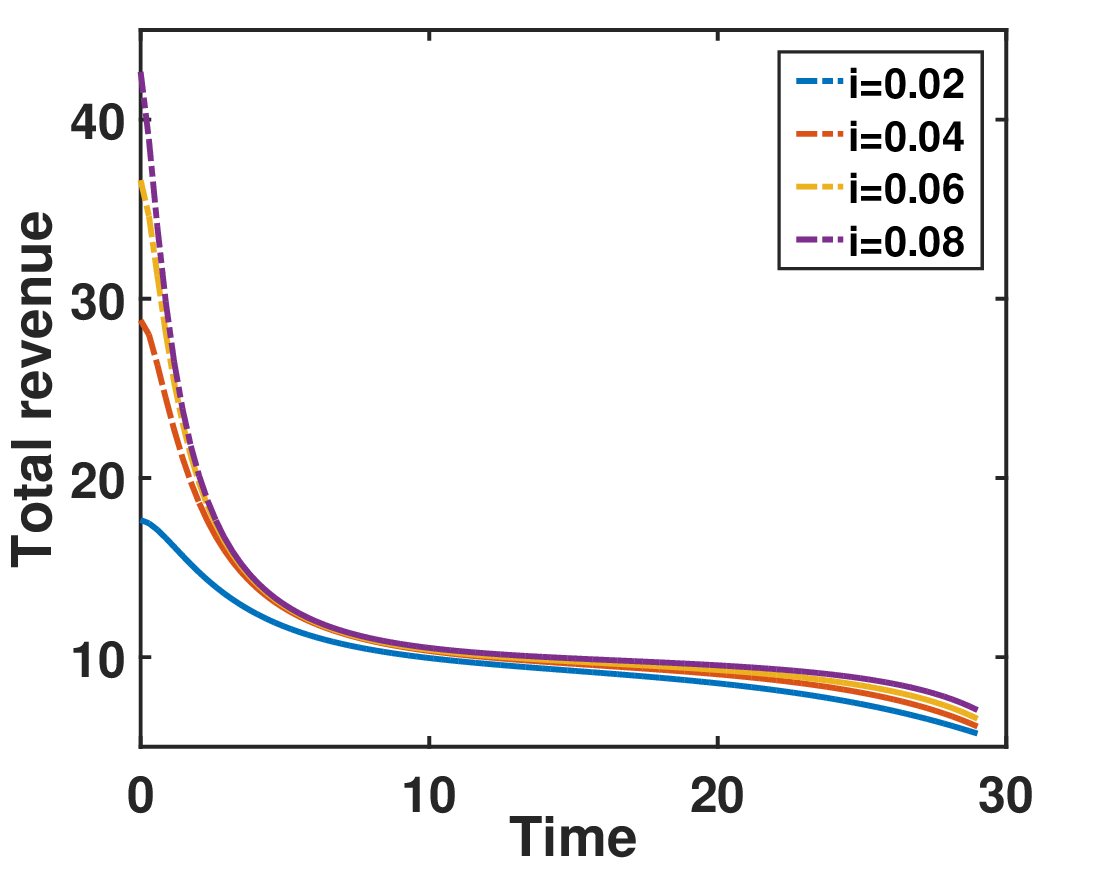}}
		\subfigure[]{\includegraphics[height=7 cm, width=8cm]{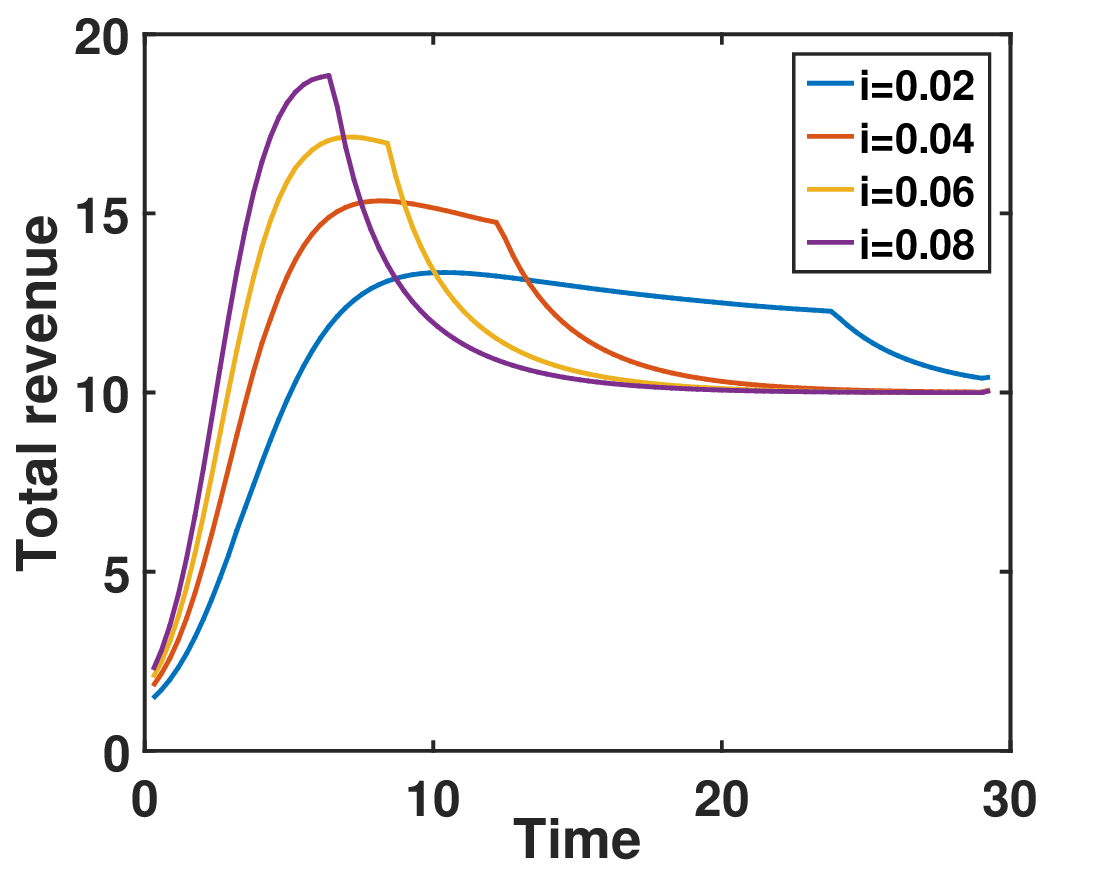}}
		\caption{These figures illustrate the net revenue produced by varying interest rates at a small intrinsic growth rate ($r=0.3$) under both the compound and annuity laws of interest.}
  \label{net_revenue_different_i_small_r}
		\end{center}
  
\end{figure}

\begin{figure}[H]
	\begin{center}
		\subfigure[]{\includegraphics[height=6 cm, width=8cm]{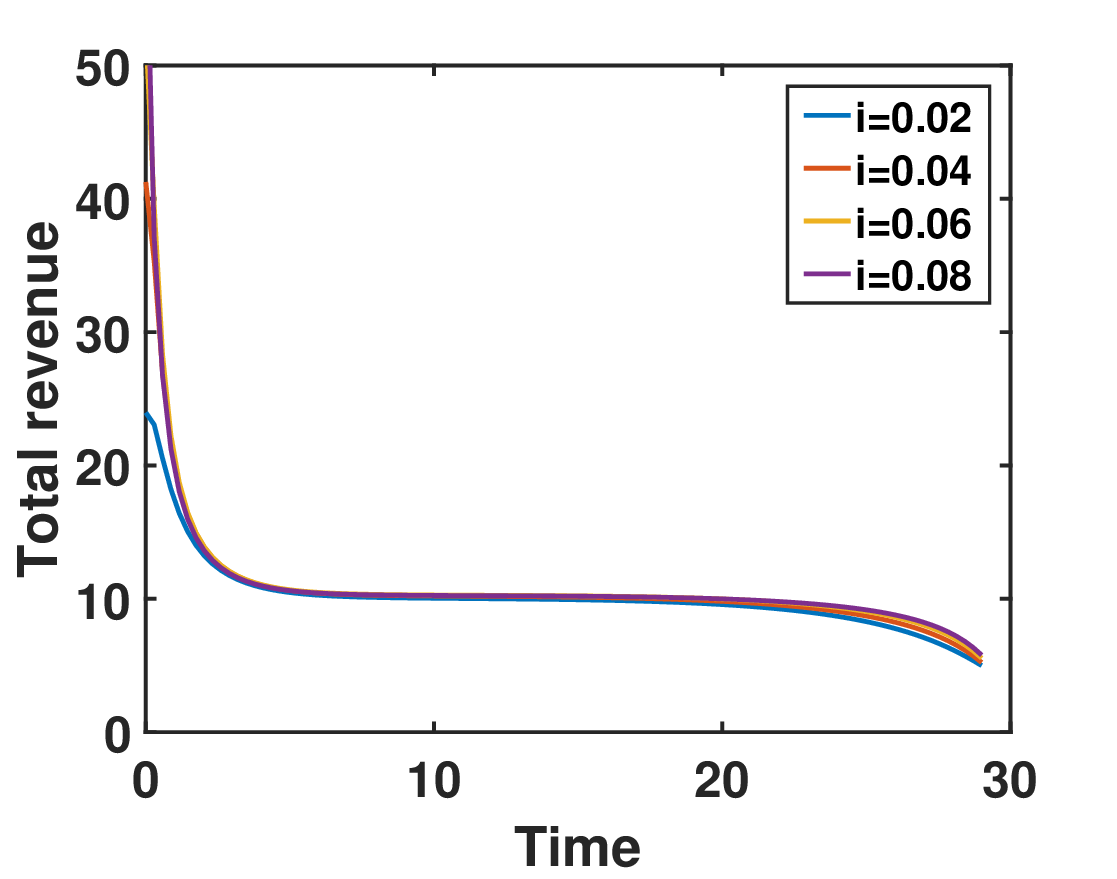}}
		\subfigure[]{\includegraphics[height=6cm, width=8cm]{total_revinue_different_i_fixed_r_annuity.eps}}\\
		\caption{This panel of figures represents the net revenue generated for different interest rates for moderate or large intrinsic growth rates ($r=0.7$) under the compound and annuity law of interest, respectively.}
  \label{net_revenue_different_i_modarate_r}
	\end{center}

\end{figure}

\begin{figure}[H]
	\begin{center}
		\subfigure[]{\includegraphics[height=7 cm, width=8cm]{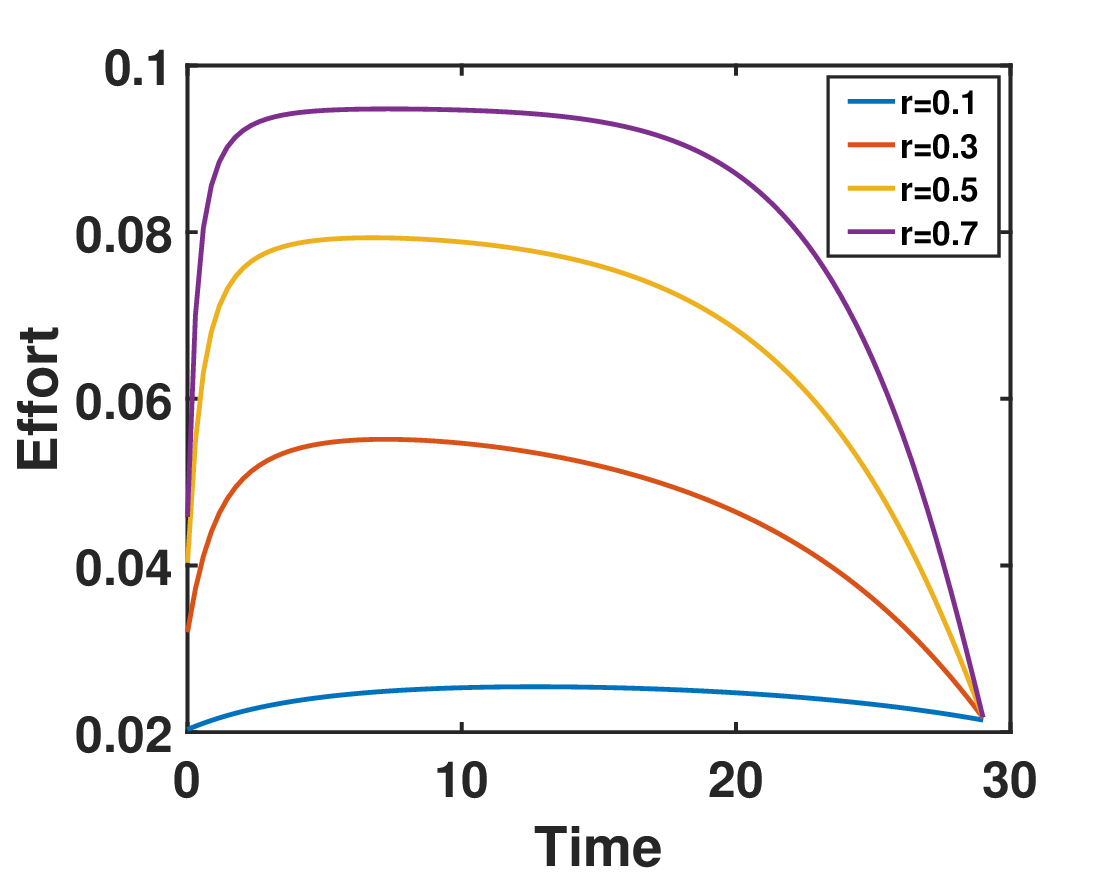}}
		\subfigure[]{\includegraphics[height=7 cm, width=8cm]{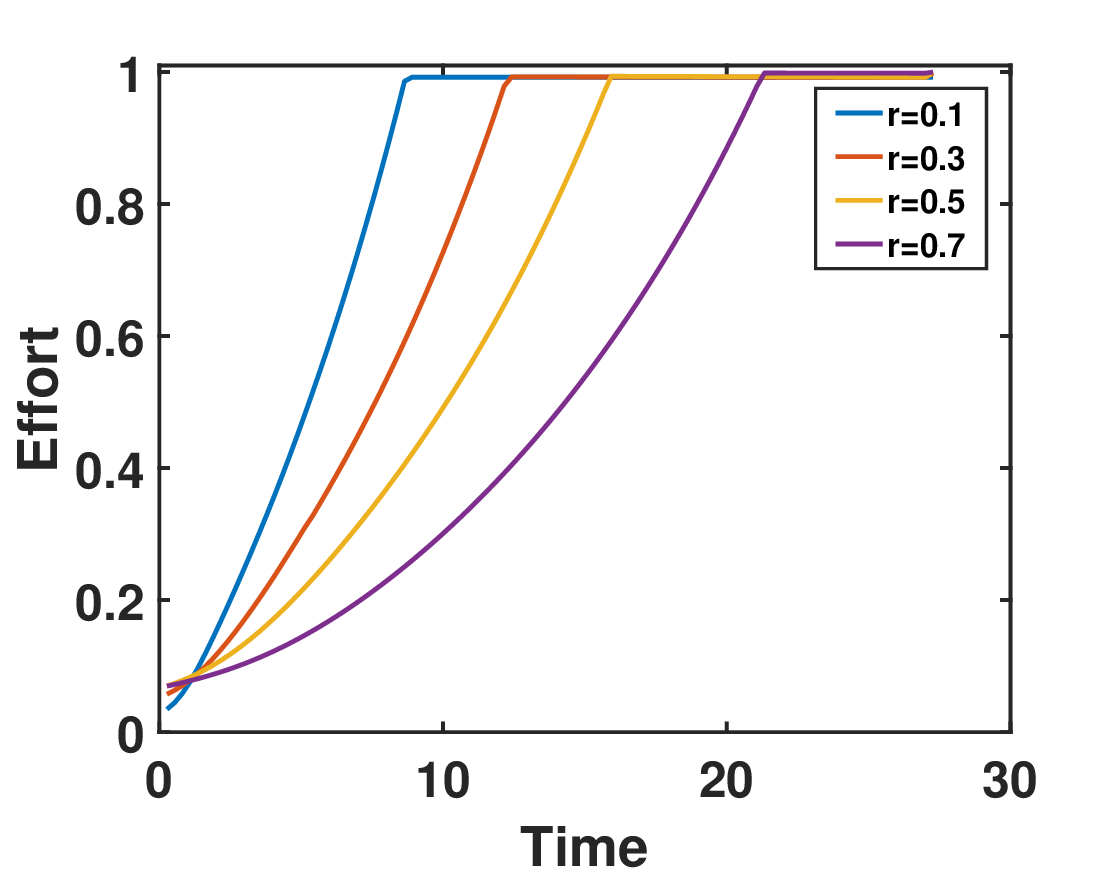}}
		\caption{These figures represent the total effort required for maximum profit under the compound and annuity law of interest, respectively, for different growth rates for the compound and annuity law of interest.}
  \label{total_effort_diff_r}
	\end{center}

\end{figure}

\begin{table}[H] 
	\scalebox{0.8}{
			\begin{tabular}{|l|l|} \hline 
				\textbf{Growth rate ($r$), (i=0.4) } &\textbf {Net revenue} \\ \hline \hline
				$r$=0.1 & 1113.7\\ 
				$r$=0.3 & 1088.9\\
			$r$=0.5 & 1082.5\\
				$r$=0.7 & 1081\\ \hline 
				
		\end{tabular}}
	\hfill
\scalebox{0.8}{
			\begin{tabular}{|l|l|} \hline 
			
				\textbf{Growth rate ($r$), (i=0.4)} & \textbf{Net revenue} \\ \hline \hline
				$r$=0.1 & 1330.1\\ 
				$r$=0.3 & 1110.9\\
			
				$r$=0.5 & 1059.5\\
				$r$=0.7 & 1034.9\\ \hline 
				
		\end{tabular}}

	\caption{Table of different values of the intrinsic growth rate when the rate of interest is fixed for the compound and annuity laws of interest.}
 \label{table_different_r}
\end{table}

\begin{table}[H] 
	\scalebox{0.8}{
	
\begin{tabular}{|l|l|}\hline

				\textbf{Rate of interest (i), ($r$=0.3) } & \textbf{Net revenue} \\ \hline \hline
				$$i=0.2$$ & 990.73\\ 
				$$i=0.4$$ & 1111\\
				$$i=0.6$$ & 1162\\
				$$i=0.8$$ & 1201.5\\\hline 
		\end{tabular}}
	\hfill
\scalebox{0.8}{	
	\begin{tabular}{|l|l|}\hline
			
				\textbf{Rate of interest (i), ($r$=0.3)} &\textbf {Net revenue} \\ \hline \hline
				$$i=0.2$$ & 1121.1\\ 
				$$i=0.4$$ & 1129.3\\
				$$i=0.6$$ & 1130.6\\
				$$i=0.8$$ & 1131.56\\ \hline 
		\end{tabular}}
	\caption{Table of different values of the rate of interest when the intrinsic growth rate is fixed (small) for the compound and annuity laws of interest.}
 \label{table_diff_i_small_r}
\end{table}

\begin{table}[H] 
	\scalebox{0.8}{

		\begin{tabular}{|l|l|}\hline
			
				\textbf{Rate of interest (i), ($r$=0.7) } & \textbf{Net revenue} \\ \hline \hline
				$$i=0.2$$ & 1025.1\\ 
				$$i=0.4$$ & 1093.3\\
				$$i=0.6$$ & 1125.8\\
				$$i=0.8$$ & 1130\\\hline 
		\end{tabular}}
\hfill
\scalebox{0.8}{

\begin{tabular}{|l|l|}\hline
			
				\textbf{Rate of interest (i), ($r$=0.7)} &\textbf {Net revenue} \\ \hline \hline
				$$i=0.2$$ & 1016.4\\ 
				$$i=0.4$$ & 1052.9\\
				$$i=0.6$$ & 1056\\
				$$i=0.8$$ & 1058.7\\\hline 
		\end{tabular}}
	
	\caption{Table of different values of the rate of interest when the intrinsic growth rate is fixed (moderate or large) for the compound and annuity law of interest.}
 \label{table_diff_i_modarate_r}
\end{table}

\begin{figure}[H]
\centering
\begin{tikzpicture}

\draw (0,0) rectangle (8,4);

\draw (.7,3.3) -- (8,3.3);
\draw (.7,2) -- (8,2);

\draw (4,0) -- (4,3.3);
\draw (.7,0) -- (.7,3.3);

\node at (4,3.7) {\bfseries Intrinsic growth rate};

\node[rotate=90] at (0.4,2) {\bfseries Interest Rate};

\node at (2,2.7) {Compound};
\node at (6,2.7) {Compound};

\node at (2,1.3) {Annuity};
\node at (6,1.3) {Compound};

\end{tikzpicture}
\caption*{This schematic diagram represents the visualization between the intrinsic growth rate and the interest law. For low interest rates and low intrinsic growth rates, the annuity law of interest provides the optimal solution; otherwise, the compound law of interest generates the net revenue.}
\end{figure}

\section{Conclusion} \label{conclusion}

This study discusses the trade-off between species growth rate and the suitable choice of optimal control policy. A comparative analysis of the compound and annuity laws of interest is presented, focusing on their relationship to species growth rate within the context of optimal control. The procedures for computing future value under both interest laws are outlined to facilitate comparison of accumulated amounts. The instantaneous discount rate for the compound law of interest remains constant, whereas it varies over time for the annuity law of interest. Our findings indicate different optimal strategies depending on varying intrinsic growth rates and interest rates. Under the annuity law of interest, a low intrinsic growth rate initially yields maximum profit. In contrast, for a compound interest, total revenue is highest at the beginning and decreases over time. For the annuity law of interest,  total revenue is stabilized when the intrinsic growth rate is increased,  but decreases beyond a certain threshold. Whereas, for the compound interest law, stabilization occurs over a longer duration, allowing for maximum profit to be sustained for an extended period. By analyzing the effort-revenue table, we concluded that fewer reproductive species give more revenue in the annuity law of interest at a low interest rate, while more reproductive species yield higher revenue under the compound law of interest. When comparing net revenue for varying intrinsic growth rates, we observed that the net revenue from the annuity interest law is greater than that from the compound interest law for a low rate of interest. However, after reaching a threshold intrinsic growth rate, the situation reverses, and higher revenue is obtained from the compound interest scheme beyond that point. 

This analysis focuses on identifying the appropriate scheme for specific biological and economic conditions. The compound interest model is recommended for species exhibiting high reproductive rates over short periods. When the intrinsic growth rate of a species is known, it can be evaluated against a predetermined threshold value. The compound law of interest is suitable if the growth rate exceeds this threshold. Conversely, for species with low intrinsic growth rates, the annuity law of interest can yield comparable benefits more rapidly, as it requires fewer time periods. Many species are harvested to meet human needs, resulting in significant economic implications.

In fisheries, stocks characterized by high reproductive rates, rapid growth, and short generation times typically exhibit a logistic growth profile. Small pelagic fishes exemplify species with high reproductive rates, whereas many deep-water fishes demonstrate lower growth rates \citep{mildenberger2025estimating, ruzicka2024role}. European anchovy (\textit{Engraulis encrasicolus}), a small pelagic fish with rapid growth and early maturity, is commonly found in European waters, Sardin (\textit{Sardina pilchardus}), a short-lived, fast-growing pelagic species, inhabits the Mediterranean Sea and the Northeast Atlantic, Jack mackerel (\textit{Trachurus spp.}) also represents a high-growth-rate species, with growth rates generally exceeding $0.6$ or $0.8$ \citep{schickele2021european}. In contrast, certain deep-water fish populations, such as blackbelly rosefish (\textit{Helicolenus dactylopterus}), Atlantic cod (\textit{Gadus morhua}), and other groundfish, experience long-term declines in productivity and are classified as low-productivity stocks, fitting a logistic model with relatively low intrinsic growth rates \citep{medeiros2025dwindling, mildenberger2025estimating}. For these species, the intrinsic growth rate is typically less than $0.3$. Thus, for small pelagic fish such as Sardin, Jack mackerel, etc, the compound law of interest yields optimal profit, whereas for Atlantic cod and other deep-water fish with low intrinsic growth rates, the annuity law of interest provides optimal profit under conditions of low interest rates.

Understanding the principles discussed in this paper is essential for making informed decisions when selecting between annuity and compound interest models under specific constraints and time frames. The comprehensive analysis presented serves as a valuable resource, guiding the selection of optimal approaches for one-dimensional models. The analogy between economic schemes and species growth, a recurring theme in this paper, provides a robust framework to enhance decision-making processes.

\bibliographystyle{plain}
\bibliography{document}

\end{document}